\newif\ifSC
\SCtrue
\ifSC
\documentclass[onecolumn,draftclsnofoot,12pt]{IEEEtran}
\else
\documentclass[10pt,conference]{IEEEtran}
\fi
\IEEEoverridecommandlockouts
\usepackage{cite}
\usepackage{amsmath,amssymb,amsthm}

\usepackage{graphicx}
\usepackage{etoolbox}
\usepackage{xcolor}
\usepackage{bbm}
\usepackage{subfigure}

\newtheorem{theorem}{Theorem}
\newtheorem{lemma}{Lemma}
\newtheorem{definition}{Definition}
\newtheorem{corollary}{Corollary}
\newtheorem{remark}{Remark}

\newcommand{\x}{\mathbf{x}}
\newcommand{\X}{\mathbf{X}}

\newcommand{\y}{\mathbf{y}}
\newcommand{\z}{\mathbf{z}}

\newcommand{\St}{\mathsf{S}}

\newcommand{\sen}{\mathrm{S}}

\newcommand{\ob}{\mathrm{o}}

\newcommand{\dv}{\mathrm{d}}

\newcommand{\B}{\mathcal{B}}

\newcommand{\C}{\mathcal{C}}

\newcommand{\R}{\mathbb{R}}
\newcommand{\set}[1]{\mathsf{#1}}

\newcommand{\ie}{{\em i.e.~}}
\newcommand{\dist}[1]{\|#1\|}

\newcommand{\tautol}[1]{\tilde{#1}}
\newcommand{\expect}[1]{\mathbb{E}\left[#1\right]}
\newcommand{\prob}[1]{\mathbb{P}\left[#1\right]}
\newcommand{\indside}[1]{\mathbbm{1}\left(#1\right)}

{}
{}
\graphicspath{{figs3/}}

\begin{document}

\title{Coverage Improvement of Wireless Sensor Networks via  Spatial Profile Information \vspace{-.1in} }

\author{\IEEEauthorblockN{Kaushlendra Pandey, Abhishek Gupta}
	\IEEEauthorblockA{\\Indian	Institute of Technology Kanpur, Kanpur (India) \\
		Email:\{kpandey,gkrabhi\}@iitk.ac.in }\vspace{-2em}\thanks{This work is supported by the
		Science and Engineering Research Board (DST, India) under the grant SRG/2019/001459.}\vspace{-.5em}}

\maketitle
\begin{abstract}
This paper considers a wireless sensor network deployed to sense an environment variable with a known spatial statistical profile. We propose to use the additional information of the spatial profile to improve the sensing range of sensors while allowing some tolerance in their sensing accuracy. We show that the use of this information improves the sensing performance of the total WSN. For this, we first derive analytical expressions for various performance metrics to measure the improvement in the sensing performance of WSN. We then discuss the sensing gains quantitatively using numerical results. 
\end{abstract}

\section{Introduction}
Modern wireless sensor networks (WSNs)  consist of a large number of inexpensive wirelessly connected sensor nodes, each enabled with a limited sensing capability.  In some applications of WSN, such as monitoring of a macro-environmental variable (MEV), we need to deploy the sensors with considerable density and monitor the area for an extended period. 
Macro-environmental variables generally have little variation over the space. For example,  {\em e.g.} environmental humidity, temperature of the earth surface observe very little change between two locations that are only meters away. Since MEVs have a spatial correlation, there must be a tremendous amount of redundancy in the sensed data when sensors are relatively dense.
Due to slow spatial variation of these variables, one can estimate their value at one location from their value at another location with high accuracy. However, we would need knowledge of the spatial profile of these variables for the estimation. The estimation would also introduce some errors, and hence, an error in the sensing accuracy needs to be tolerated while using the spatial correlation information. \par  
In the past literature, the performance of WSNs has been studied both numerically via simulations and analytically using tools from  stochastic geometry \cite{AndGupDhi16}. In \cite{iyer2008limit} authors presented a model for the $k$-coverage for Boolean-Poisson model for the intrusion detection.  A survey on coverage control algorithms, along with the relation of coverage and connectivity, is available in \cite{more2017survey}. A comprehensive survey on directional and barrier coverage is available in \cite{wu2016survey}. Comparative coverage analysis for WSNs when sensors’ locations are according to different point processes is performed in \cite{pandey2020coverage}.  The researchers have also studied various ways to increase the coverage region, including the densification of sensors, optimal deployment, and increasing the sensor range.  For example, in \cite{pandey2019detection} focuses on the early detection of a forest fire. In \cite{tan2012exploiting} authors showed that sensing coverage performance can be improved by using data fusion techniques where each sensor sends its measurements about the sensed signal to the cluster head, which makes the detection decision based on the received measurements. One way to increase the coverage of WSN without deploying more sensors is to avoid redundancy in the sensed data.    In WSNs that are deployed to sense a target MEV, the spatial variation profile of the MEV may help us increase coverage by estimating the value of MEV  in the uncovered regions. 
The spatial variation of MEVs has been studied in the past works. For example, in \cite{kaushik1965preliminary}, authors present an approximate relation between the temperature of the soil and its depth. Similarly, the numerical study of variation of forest temperature over space and time was performed in \cite{kawanishi1986numerical}.   It is interesting to study if the spatial profile information can be used to achieve better coverage of a region without increasing the sensor density by avoiding redundancy in the sensed data. There are only a few past works that focus on the spatial correlation of variables.   In \cite{liu2007energy}, clusters of closely located sensors are formed as they have similarity in the sensed data. For cluster formation, the sink node observes the reading of source nodes for a time-period, and based on the observation, it creates the clusters.  In \cite{jarwan2019data}, the authors used spatio-temporal correlation among the sensed data for dual prediction and data compression where sensors predict data based on past observations. However, a comprehensive analytical framework for the sensing performance of a WSN deployed to sense an environmental variable with known spatial profile information, is not still available which is the main focus of this paper. \par
In this paper, we consider a WSN deployed to sense an environment variable with a known spatial statistical profile. We propose to use the additional information of the spatial profile to estimate the value of the environmental variable in the uncovered regions from the value at covered regions. This can improve the sensing range of sensors while introducing some estimation errors in the sensing accuracy which need to be tolerated. We first derive analytical expressions for various performance metrics to measure the improvement in the sensing performance of WSN. Using these analytical expressions and quantitative results, we show that the use of this information improves the WSN's sensing performance.\par
\textbf{Notation:}  $\B(\y,a)$ denotes a $2$-$d$ ball of radius $a$ centered at location $\y$.  $\X_i$ denotes the location of $i$-th sensor. Let $R_\sen$ be the sensing range of each sensor. Hence, $\St\equiv\B(\ob,R_\sen)$ denotes the sensing region of the sensor located at the origin $\ob$. Let $\C\equiv\B(\ob,r)$ be the region of interest. Let $|\mathsf{A}|$ denote the Lebesgue measure of the set $\mathsf{A}$. The Minkowski sum of any two set $\mathsf{A}\oplus\mathsf{B}$ is defined as $\{a+b: a\in\mathsf{A},b\in\mathsf{B}\}$. The Minikowski difference  of the two set $\mathsf{A}\ominus\mathsf{B}=(\mathsf{A^{c}}\oplus\mathsf{B})^{\mathsf{c}}$. 
\section{System model}
In this paper, we consider a  WSN deployed in $\R^2$  to sense an environmental variable $\Theta$ that varies spatially. Examples include the temperature in a forest, soil moisture in an agricultural field, humidity in a city. The locations of sensors can be modeled as homogeneous Poisson point process (PPP) $\Psi=\{\X_i \}$ with intensity $\lambda$ \cite{AndGupDhi16}. We assume that each sensor $\X_i$ has a circular sensing region $\St_i=\B(\X_i,R_\sen)$ around it. The  region sensed by WSN is Boolean-Poisson model $\xi$ given as
\begin{align}
\xi=\bigcup_{\X_i\in\Psi}\X_i+\St,\label{eq:sensedarea}
\end{align}
where $\St\equiv\B(\ob,R_\sen)$ denotes the sensing region of each sensor around itself. 
\subsection{Profiling of spatial variation of the environmental variable}
Let the value of the environmental variable at a location $\x$ be denoted by $\Theta(\x)$.
For real-world cases it can be assumed that the spatial variation of $\Theta$ is bounded which means that it can vary only by a finite value in a finite distance.  One example  can be found in \cite{biswas1979estimation} where the spatial profile of the soil moisture $S$ with depth $y$ is given as 
\begin{align*}
S(y)=A(y)+S(0)[1+B(y)^2]+S_c.
\end{align*}
Here, $S(y)$ is the soil moisture at the depth $y$, $S(0)$ the soil moisture at or near the surface layer at depth $0$, and $A$, $B$ and $S_c$ are some constants. If $S(x)$, $y$ and $x$ are known, $S$ can be estimated.  

This assumption results in  variable having a spatial profile, and the spatial correlation at the variable's value at two points. We assume the knowledge of this spatial profile.
In particular, we assume the following spatial profile that for any two points  $\x$ and $\y$, the variation in the value of $\Theta$ 
\begin{align}\label{parametervari}
|\Theta(\x)-\Theta(\y)|\leq f(\dist{\x-\y},w),
\end{align} 
where $f(\dist{\x-\y},w)$ is a tolerance function and $w$ is the spatial variation rate of  $\Theta$.  Hence, the uncertainity in the value of $\Theta(\y)$ conditioned on the knowledge of $\Theta(\x)$ is 
\[\mathsf{Uncert}(\Theta(\y)|\Theta(\x))=f(\dist{\x-\y},w).\]
\subsection{Use of the environmental variable's spatial profile}
Due to this additional information about the correlation, the variable's value at  locations that are not covered in the sensing range of any sensor, can be guessed/estimated within some tolerance as shown in Lemma \ref{thm1}. 
\begin{remark}
$f$ is an increasing function with respect to the first argument which can be seen as follows. If $\y$ is close to $\x$, $\Theta(\y)$ is equal to $\Theta(\x)$ (\ie $f(0,w)=0$), In other words, $\Theta(\y)$ can be predicted exactly.
 As we move $\y$ away from $\x$, the correlation between $\Theta(\x)$ and $\Theta(\y)$ will reduce and the certainty in the value of $\Theta$ decreases. 
\end{remark}
\begin{lemma}\label{thm1}
	If the value of $\Theta$ at $\x$ (\ie  ~ $\Theta(\x)$) is known, then the set of points where uncertainty in $\Theta$ is within $\tau$ tolerance, is given as
	\[
	\mathsf{P}(\x,\tau)=
	\left \{\y:|\Theta(\y)-\Theta(\x)|<\tau\right\}=\B(\x,R(\tau,w)),
	\]
	where $R(\tau,w)$ is given as
	\begin{equation}
	R(\tau,w) =f^{-1}(\tau,w).                             
	\end{equation} 
	Here, inverse of $f$ is with respect to the first argument.
\end{lemma}
\begin{proof}
To ensure that $\Theta(\y)$ does not vary more than $\tau$ from $\Theta(\x)$, $f(\dist{\x-\y})$ must satisfy
\begin{align*}
f(\dist{\x-\y},w)\leq\tau \implies
\dist{\x-\y}\leq f^{-1}(\tau,w).
\end{align*}
\end{proof}

\begin{lemma}\label{thm2}
	If the value of $\Theta$ at all points in a set $\set{A}$ is known,  the set of points where $\Theta$ can be predicted within $\tau$ tolerance, is given as
	\begin{align*}
	\mathsf{P}(\set{A},\tau)&=
	\left \{\y:|\Theta(\y)-\Theta(\x)|<\tau,\text{for at least a point }\x\in\set{A}\right\}\\&=\set{A}\oplus \B(\ob,R(\tau,w)).
	\end{align*}
\end{lemma}
\begin{proof}
Owing to increasing nature of $f$, it is best to use the closest point in $\set{A}$ to estimate the value of $\Theta(\y)$ at a location $\y$. Hence,
	\begin{align*}
	\mathsf{P}(\set{A},\tau)&=\cup_{\x \in \set{A}}\mathsf{P}(\x,\tau)=\set{A}\oplus \B(0,R(\tau,w)).
	\end{align*}\vspace{-.45in}\ \\
\end{proof}

\subsection{$\tau-$tolerance sensed region}
From \eqref{eq:sensedarea}, $\xi$ is the region where value of $\Theta$ is exactly known. From Lemma \ref{thm2},  the region in which $\Theta$ can be sensed within $\tau$ tolerance is given as
\begin{align}
\tautol{\xi}=\xi\oplus \B(\ob,R(\tau,w))=\bigcup_{\X_i\in\Psi}\X_i+\tautol{\St},
\end{align}
where $\tautol{\St}=\B(\ob,R_\sen)\oplus\B(\ob,R(\tau,w))=\B(\ob,R_\sen+R(\tau,w))$.  We term $\tautol{\xi}$ as {\em $\tau-$tolerance sensed region}  and $\tautol{\St}_i=\X_i+\tautol{\St}$ as {\em $\tau-$tolerance sensing zone} of the sensor at $\X_i$. Note that $\tautol{\xi}$ is also a Boolean Poisson model. Note that $\xi$ which is the exact sensed area, can be obtained  from $\tautol{\xi}$ by substituting $\tau=0$. Hence, $\xi$ (\ie~ the 0-tolerance sensed area) may be seen as a special case of $\tautol{\xi}$. 
\par We now analyze the sensing and covering performance of WSN. We give the following two definitions for their use in next sections.
\begin{definition}
	A point $\z\in\R^{2}$ is said to be $\tautol{\xi}-$ covered if $\z$ falls in the $\tau-$tolerance sensing zone of at least one sensor \ie  $\z\in\tautol{\xi}$.
\end{definition}
\begin{definition}
	A point $\z\in\R^{2}$ is said to be $\tautol{\xi}-$ covered by exactly $k$ sensor if $\z$ falls in the $\tau-$tolerance sensing zones of exactly $k$ sensors.
\end{definition}

\section{Sensing Performance Analysis}
In this section, we analyze the sensing performance of the WSN. Let $\C$ be a set denoting a region of interest.
 \subsection{$m-$sensed area fraction}
We will first derive the $\tau-$tolerance at-most-$m$-sensed area fraction ($\nu_{m}(\tau)$) which is defined as average fraction of  $\C$ falling under the $\tau-$tolerance sensing region of at-most $m$ sensors \ie
	\begin{align}\label{expressionVm(c)}
	\nu_{m}(\tau)=\frac1{|\C|}\mathbb{E}\left[\sum_{k=1}^{m}\int_{\C}\gamma_{k}(\z)\dv \z\right],
	\end{align}
	where $\gamma_{k}(\z)$ is defined as
	\begin{equation*}
	\gamma_{k}(\z)=\indside{\z \text{ is $\tautol{\xi}-$covered by exactly $k$ sensors}}.
	\end{equation*}
\begin{lemma}\label{lemma3}
The probability that a point $\y\in \R^{2}$ is  $\tautol{\xi}-$covered by exact $k$ sensors  is equal to
\begin{align*}
\prob{\gamma_k(\z)}=\exp\left(-\lambda \pi R^{2}_\sen(\tau)\right)\frac{\left(-\lambda \pi R^{2}_\sen(\tau)\right)^k}{k!},
\end{align*}
	where $R_{\sen}(\tau)=R_\sen+R(\tau,w)$. 
\end{lemma}
\begin{proof}
$\z$ will fall in $\tau-$tolerance sensing zone of $k$ sensors if there are exactly $k$ sensors in the $\B(\y,R_{\sen}(\tau))$. Since sensors follow PPP, we get the desired result.
\end{proof}
Applying Lemma \ref{lemma3} in \eqref{expressionVm(c)}, we get the following theorem.
\begin{theorem}\label{thm:a1}
$\tau-$tolerance at-most-$m$-sensed area fraction $\nu_{m}(\tau)$ of the WSN is given as
	\begin{align*}
	&\nu_{m}(\tau)=\sum_{k=1}^{m}e^{-\lambda \pi R^2_{\sen}(\tau)}\frac{(\lambda\pi R^2_{\sen}(\tau))^k}{k!}.
	\end{align*}
\end{theorem}
Note that the above result is independent of the set $\C$ owing to the stationarity of the WSN \cite{chiu2013stochastic}.
\subsection{Sensed area fraction}
We now derive the $\tau-$tolerance  sensed area fraction, also termed $\tau-$SAF, ($\nu(\tau)$) which is defined as the fraction of a set $\C$ that can be sensed by $\xi$ within $\tau-$ tolerance. Mathematically, it is equal to
\[	\nu(\tau)=\frac{\expect{|\tautol{\xi}\cap\C|}}{|\C|}\]
and it can be expressed in terms of $\nu_{m}(\tau)$ as
\begin{align}\label{eq:nutaunumtaurelation}
\nu(\tau)=\lim_{m\rightarrow\infty}\nu_{m}(\tau).
\end{align}
\begin{corollary}
	The $\tau-$SAF $\nu(\tau)$ is given as
	\begin{align*}
	\nu(\tau)=(1-\exp\left(-\lambda\pi R^2_{\sen}(\tau)\right)).
	\end{align*}
\end{corollary}
\begin{proof}
The result can be  obtained Theorem \ref{thm:a1} and \eqref{eq:nutaunumtaurelation}.
	\end{proof}
\begin{corollary}
$\nu_\mathrm{o}(\tau)=1-\nu(\tau)=\exp\left(-\lambda \pi R^2_{\sen}(\tau)\right)$, represents the average vacant fraction.
\end{corollary}
The gain due to the use of additional correlation information can be expressed in terms of the coverage improvement factor (CIF) which is defined as the relative improvement in $\nu$ when allowing tolerance in sensing with the use of spatial correlation information. Mathematically,
\begin{align*}
\eta(\tau)&=\frac{\nu(\tau)}{\nu(0)},
\end{align*}
\vspace{-1.1em}
and is given as
\vspace{-.1em}
\begin{align*}
\eta(\tau)&=\frac{1-\exp(-\lambda\pi R^2_\sen(\tau))}{1-\exp(-\lambda\pi R_\sen^2)}.
	\end{align*}

\newcommand{\mintersectprob}{\mu_m(\tau)}
\newcommand{\intersectprob}{\mu(\tau)}

\newcommand{\mcoverprob}{\beta_m(\tau)}
\newcommand{\coverprob}{\beta(\tau)}
Now, we focus on how well a network of sensors can cover a region of interest $\C$. Let  $\C\equiv \B(\ob,r)$.
\subsection{$m-$intersection probability}
$\tau$-tolerance $m$-intersection probability ($\mintersectprob$)
   is defined as the probability that $\C$ has non-empty intersection with the $\tau-$tolerance sensing zone of exactly $m$ sensors. 
The $\tau-$tolerance sensing zone of a sensor intersects with $\C$ if and only if the location of sensor falls in the  Minkowski sum of $\C$ and $\tautol{\St}$ \cite{haenggibook}. 
   Hence $\mintersectprob$ is equal to the probability that there are exactly $m$ sensors in $\C\oplus\tautol{\St}$ \ie 
	\begin{align*}
	\mintersectprob=\mathbb{P}\left[\Psi(\tautol{\St}\oplus\C)=m\right].
	\end{align*}
	Now, $\C=\B(\ob,r)$. Noting that $\tautol{\St}\oplus\C=\B(\ob,R_\sen(\tau))\oplus\B(\ob,r)=\B(\ob,R_\sen(\tau)+r)$, we get 	
	\begin{align*}
	\mintersectprob=\mathbb{P}\left[\Psi(\B(\ob,R_\sen(\tau)+r))=m\right].
	\end{align*}
	Now, since $\Psi$ is a PPP, we get the following theorem.
\begin{theorem}\label{thm2}
	The $\tau$-tolerance $m$-intersection probability $\mintersectprob$ for $\C=\B(\ob,r)$ is
	\begin{align}\label{m_cov_FHPPP}
	\mintersectprob&=e^{-\lambda\pi(R_\sen(\tau)+r)^2 }\frac{\left(\lambda\pi(R_\sen(\tau)+r)^2\right)^m}{m!}.
	\end{align}
\end{theorem}

\begin{remark}\label{thm:4}
	The density $\lambda^{\mathrm{opt}}$ that maximizes the $\tau$-tolerance $m$-intersection probability is equal to
\[\lambda^{\mathrm{opt}}=\frac{m}{\pi (R_\sen(\tau)+r)^2}.\]
	and the maximum value of $\mintersectprob$ is
	\begin{align}
	\mintersectprob_\mathrm{max}={me^{-m}}/{m!}.
	\end{align}
\end{remark}
There are applications where sensors cooperatively decide the value of the environmental variable and there may a minimum limit on number of sensors require to build a consensus. 
The $m$-sensed area fraction and $m$-intersection probability are useful metrics for these cases. 
Maximizing the  metrics $\mintersectprob$ or $\mcoverprob$ for a particular value of $m$ may help optimizing the performance of network to build the optimal consensus among sensors. 


\subsection{Intersection probability}
We now derive the $\tau$-tolerance intersection probability ($\intersectprob$) which is defined as the probability that $\C$ has non-empty intersection with the $\tau$-tolerance sensing zone of at least one sensor \ie
\[\intersectprob=\prob{\Psi(\tautol{\St}\oplus\C)\ge 1}=\sum_{m=1}^\infty \mintersectprob.\]
\begin{corollary}\label{one-intersection-prob}
The $\tau$-tolerance intersection probability for $\C=\B(0,r)$ is
	\begin{align}\label{at_least_1_cover}
\intersectprob &=1-e^{-\pi \lambda (r+R_\sen(\tau))^{2}}.
	\end{align}
\end{corollary}
\subsection{m-cover probability}
We now derive the
$\tau$-tolerance $m$-cover probability $\mcoverprob$ which is defined as the probability that $\C$ lies entirely inside the $\tau$-tolerance sensing zone of exactly $m$ sensors \ie
\begin{align}
\mcoverprob&=\prob{\left(\sum_{\X_i\in\Psi} \indside{\C \subset \X_i+\tautol{\St}}\right)=m}.
\end{align}
Now, the location of the sensors that can fully cover $\C$ are the one inside the Minkowski difference $\tautol{\St}\ominus\C$.  
Hence, $\mcoverprob$ is equal to the probability that there are exactly $m$ sensors in $\tautol{\St}\ominus\C$ \ie 
	\begin{align*}
	\mcoverprob=\mathbb{P}\left[\Psi(\tautol{\St}\ominus\C)=m\right].
	\end{align*}
	Now, $\C=\B(\ob,r)$. Note that $\tautol{\St}\ominus\C=\B(\ob,R_\sen(\tau))\ominus\B(\ob,r). $ This is equal to $\B(\ob,R_\sen(\tau)-r)$ if $R_\sen(\tau)>r$, otherwise is equal to the null set $\phi$. Substituting the value, we get 	
	\begin{align*}
	\mintersectprob=\mathbb{P}\left[\Psi(\B(\ob,R_\sen(\tau)-r))=m\right].
	\end{align*}
	Now, since $\Psi$ is a PPP, we get the following theorem.
\begin{theorem}
The $\tau$-tolerance $m$-cover probability $\mcoverprob$ is
	\begin{align}\label{m_cov_FPPP}
	\mcoverprob =&\indside{R_\sen(\tau)>r} 
	\exp{\left(-\lambda  \pi{\left(R_\sen(\tau)-r\right)}^2\right)}\times 
	\nonumber\\&\ \ \frac{(\lambda\pi {\left(R_\sen(\tau)-r\right)}^2 )^m}{m!}.
	\end{align}
	Note that a sensor can only cover $\B(\ob,r)$ if $R_\sen(\tau)>r$.
\end{theorem}
%
%
	
%
%

\subsection{Cover probability}
We now derive the $\tau$-tolerance cover probability ($\coverprob$) which is defined as the probability that $\C$ lies entirely inside the $\tau$-tolerance sensing zone of at least one sensor \ie
\[\coverprob=\prob{\Psi(\tautol{\St}\ominus\C)\ge 1}=\sum_{m=1}^\infty \mcoverprob.\]
\begin{corollary}\label{one-cover-prob}
The $\tau$-tolerance cover probability for $\C=\B(0,r)$ is
	\begin{align}\label{at_least_1_cover}
\coverprob &=\indside{R_\sen(\tau)>r}  (1-e^{-\pi \lambda (R_\sen(\tau)-r)^{2}}).
	\end{align}
\end{corollary}

%


\section{Numerical Results}
In this section,  we will numerically evaluate the presented analysis to find the impact of different system parameters on the performance of WSN. For numerical simulation purposes, we  consider the following  tolerance function $f(\cdot,\cdot)$:
\begin{align}\label{envoirnmnetalvari}
f(\dist{\x-\y},w)=\begin{cases}
Ae^{w\dist{\x-\y}}& \text{if } \x \ne \y,\\
0 & \text{if } \x=\y
\end{cases}
\end{align} 
where $w>0$, is the spatial variation rate of  $\Theta$. This form of $f$ is inspired by the relation of soil temperature at the surface and at a given depth presented in \cite{kaushik1965preliminary}. 
For this function, $\tau$-tolerance sensing radius $R_\sen(\tau,w)$ is given as
\begin{align}
	R_\sen(\tau,w) =
	\begin{cases}   
	\ln(\tau/A)/w & \text{if }  \tau>A  \\ 
	0                  & \text{if } 0\le \tau\le A                          
	\end{cases}      . 
\end{align} 
Note that lower value of $w$ corresponds to slow variation of $\Theta$ and hence large region can be sensed within a certain tolerance. On the other hand, higher value of $w$ corresponds to fast variation of $\Theta$ resulting in a high uncertainty in sensed data beyond the sensing range of $R_\sen$. The  case $w=\infty$ corresponds to complete uncertainty beyond $R_\sen$ and hence, $R_\sen(\tau,\infty)=0$  resulting in no additional coverage.\\
\noindent\textbf{{{{Impact of sensor density on $\tau$-SAF: }}}}
Fig. \ref{FAC_ppp} shows the variation of $\tau$-SAF with the sensor density $\lambda$. The maximum value of $\tau$-SAF is 1 (corresponding to the case when $\C\subset \tautol{\xi}$). It can be observed that increasing $\lambda$ increases both $\nu(\tau)$ and $\nu(0)$, however, the relative gain decreases with $\lambda$ beyond a certain density.  When sensing tolerance is allowed, for moderate value of $\lambda$, the use of spatial profile information can give us significant gain (up to 76\% for $w=.01$ in Fig. \ref{FAC_ppp}) in the average sensed area. If we fix the target SAF at a certain value (for example, $\gamma=0.8$), the required density to achieve this target SAF can be reduced (82/km$^2$ to 8/km$^2$ for $w=.01$ in Fig. \ref{FAC_ppp}) owing to the reduction in data redundancy. If spatial variation is slower, a less sensor density would be required due to increased redundancy.

\begin{figure}[ht!]
	\def\svgwidth{\textwidth}
	\centering
	\includegraphics[width=.5\textwidth]{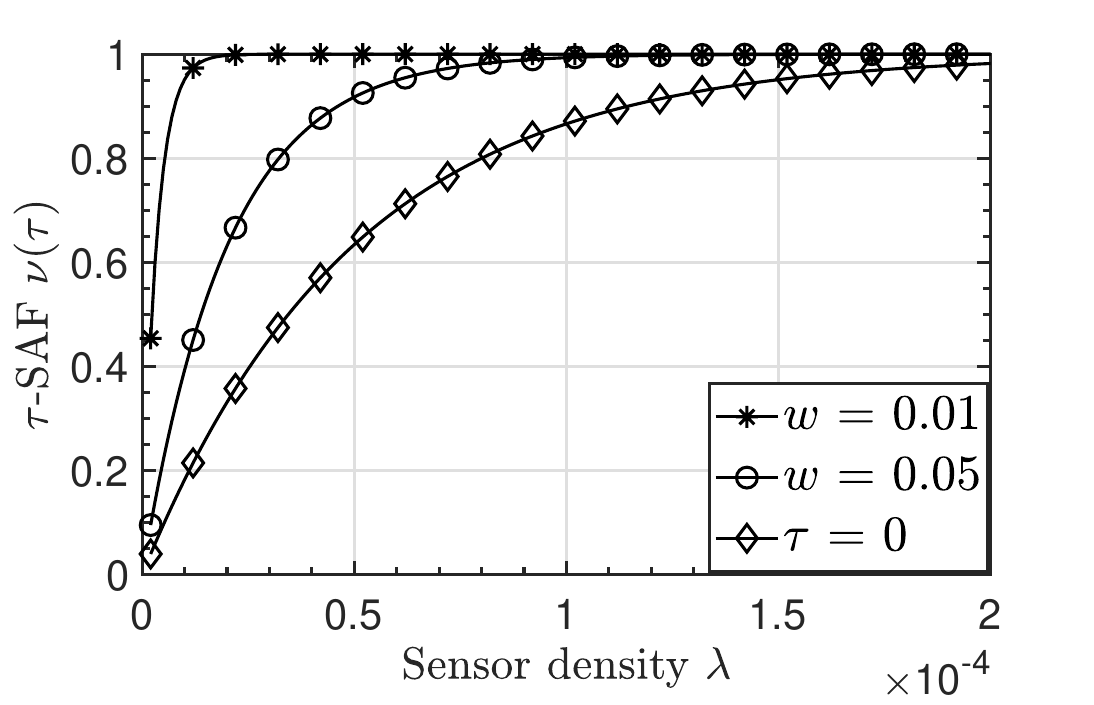}
	\caption{$\tau$-tolerance sensed area fraction for WSN with $R_\sen=80$. Here $A=1$ and allowed tolerance $\tau=10$. Allowing tolerance can reduce the required density to achieve the same target SAF, which can help reducing cost and improve the lifetime of WSN.} \label{FAC_ppp}
	\centering
	\includegraphics[width=.5\textwidth]{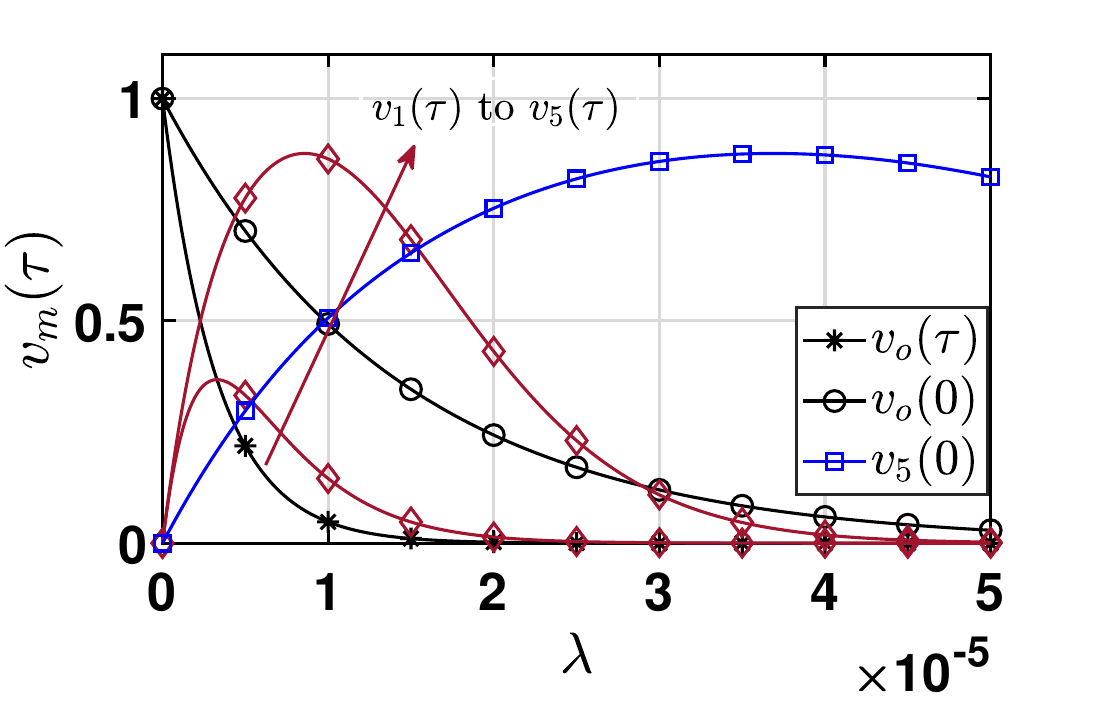}
	\caption{$\tau-$tolerance at-most-$m$-sensed area fraction and vacancy $\nu_\ob(\tau)$ for a WSN with $\tau=5, w=.01, R_\sen=150, A=1$. 
		$\nu_\ob(\tau)$  reduces considerably with spatial correlation information.} \label{m_vacancy}	
\end{figure}
\begin{figure}[ht!]
	\def\svgwidth{\textwidth}
	\centering
	\includegraphics[width=.5\textwidth]{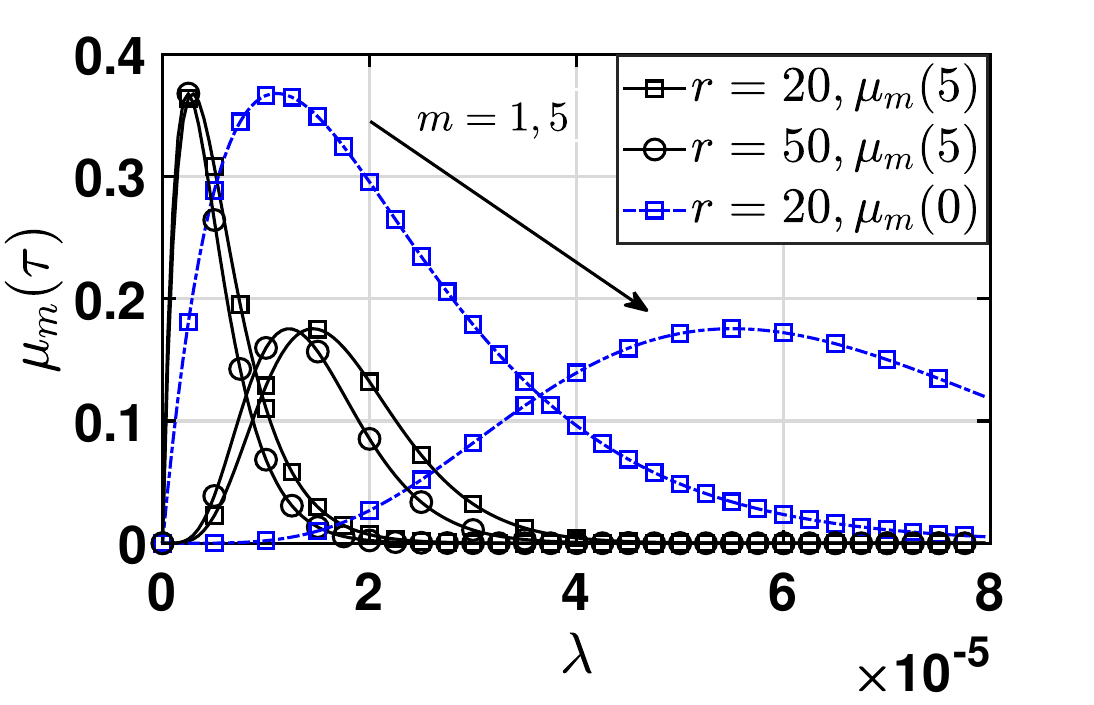}
	\caption{$m$-intersection probability $\mintersectprob$ of  a WSN with $w=.01, R_\sen=150$ for $\C=\B(0,r)$. The maximum value of $\mintersectprob$ does not depend on the size of $\C$.} \label{m_cover_ppp}
	\def\svgwidth{\textwidth}
	\centering
	\includegraphics[width=.5\textwidth]{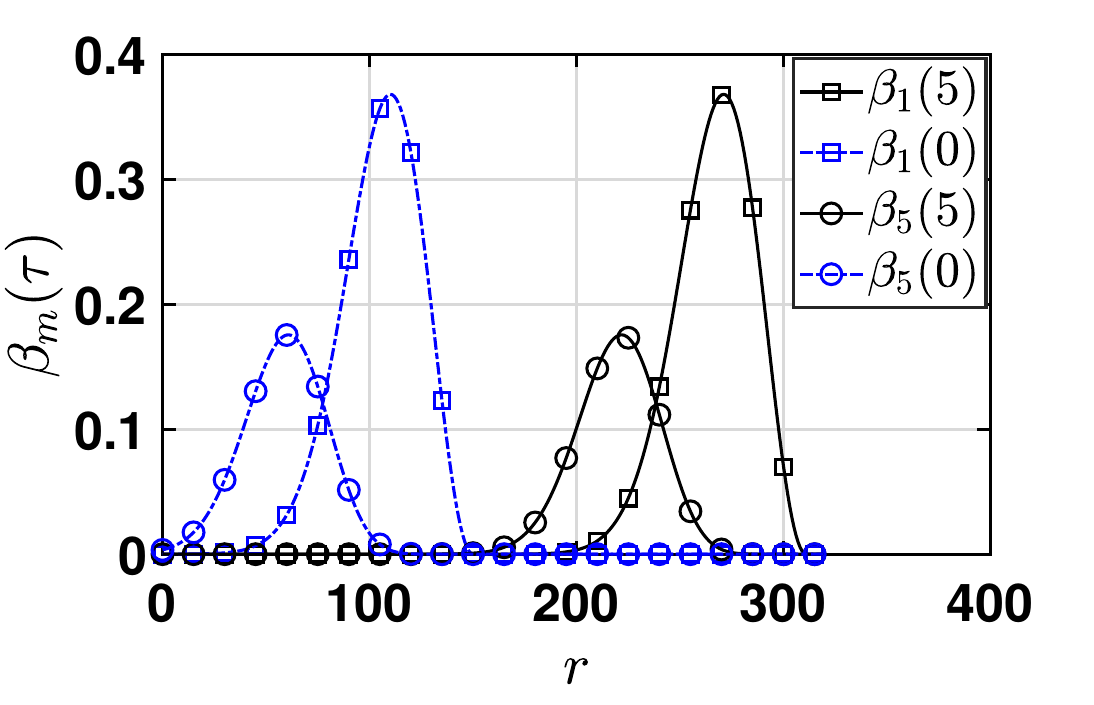}
	\caption{$m$-cover probability of a WSN with $w=.01, R_\sen=150$ for set $\C=\B(0,r)$.  With use of spatial information, it is possible to cover sets with higher radius.} \label{m_cover_ppp_F}
\end{figure}

\noindent\textbf{Impact of the sensor density on $\tau$-tolerance at-most-$m$-sensed area $\nu_{m}(\tau)$ and vacancy $\nu_\mathrm{o}(\tau)$}:
Fig. \ref{m_vacancy} shows the variation $\nu_{m}(\tau)$ with the  sensor density $\lambda$. It can be observed that $\nu_{m}(\tau)$ increases with $m$ (consistent with \eqref{expressionVm(c)}). As we increase $\lambda$, $\nu_m(\tau)$ first increases. However, there is an optimal density beyond which, $\nu_m(\tau)$ starts decreasing. This is because  an increase in the density beyond a certain point is not helpful,  due to increased redundancy in sensing data.  However, with spatial correlation this redundancy is reached beforehand which can be seen when comparing the optimal density value for the curves $\nu_5(\tau)$ and $\nu_5(0)$. 
Fig. \ref{m_vacancy} also shows the variation of average vacancy $\nu_\ob(\tau)$ with  $\lambda$. It can be seen that the vacancy decreases with $\lambda$. However, when spatial information is used, vacancy is reduced. 


\noindent\textbf{{Effect of the sensor density on $m$-intersection probability:}} 
Fig. \ref{m_cover_ppp} shows the variation of $m$-intersection probability with  $\lambda$. It can be observed that there exists an optimal value $\lambda^{\mathrm{opt}}$ of density for a particular value of $m$. Below $\lambda^{\mathrm{opt}}$,  $\mintersectprob$ is less as there are not enough sensors to intersect $\C$. When density  is higher than $\lambda^{\mathrm{opt}}$, there are more than required sensors which creates redundancy and therefore, $\mintersectprob$ is also less for this case.  Sensors can  decide which $m$ can be sufficient for the coverage of $\C$. In the case when there are more number of sensors than required, some sensors may choose to turn-off themselves. It can be observed that the maximum value of $\mintersectprob$ does not depend on   the radius of $\C$. The value of $r$ only shifts the value of $\lambda^{\mathrm{opt}}$.  Increasing $r$ or tolerance $\tau$ decreases the value of $\lambda^{\mathrm{opt}}$ that can also be seen from Remark \ref{thm:4}. This is due to the fact that data redundancy is reached earlier when spatial correlation is used or when sensors with larger sensing range are used.\\
\noindent\textbf{{Variation of  $m$-cover probability with the size of $\C$:}} 
Fig. \ref{m_cover_ppp_F} shows the variation of $m$-cover probability with  the radius $r$ of set $\C$. It can be observed that it is not possible to cover $\C$ when $r>R_\sen$ with 0 tolerance. However, with use of spatial information, it is possible to cover sets with higher radius (200-300 in this example) with some tolerance.
\section{Conclusion}
In this paper, we consider a WSN deployed to sense a target environment variable. We showed that  the sensing performance of a WSN can be improved using the information on spatial profile of the target environment variable while allowing some error tolerance in the sensing accuracy.  We also saw that the required density to achieve a certain sensing performance can be reduced when  the information on spatial profile of the target environment variable is used. This can help us reduce the cost of network, both- capital (by deploying less sensors) and operating (by keeping some sensors off to save their life-time) costs.	
\bibliographystyle{ieeetran}

\vspace{12pt}
\color{red}

\end{document}